\documentclass[10pt,letterpaper]{article}
\usepackage{etoolbox}
\usepackage{amsmath,amssymb,amsbsy}
\ifdef{\ClassIOPJournal}{}{\usepackage[margin=1in]{geometry}}
\usepackage{amsthm}
\usepackage{indentfirst}
\usepackage{xcolor}
\usepackage{graphicx}
\usepackage{microtype}
\usepackage{xurl}
\theoremstyle{plain}
\newtheorem{theorem}{Theorem}[section]
\newtheorem{lemma}[theorem]{Lemma}
\newtheorem{corollary}[theorem]{Corollary}
\theoremstyle{definition}
\newtheorem{definition}[theorem]{Definition}
\theoremstyle{remark}
\newtheorem{remark}[theorem]{Remark}

\newcommand{\dvec}{\boldsymbol{d}}
\newcommand{\nhat}{\hat{\mathbf{n}}}
\newcommand{\nz}{\hat{\mathbf{n}}_0}
\newcommand{\uhat}{\hat{\mathbf{u}}}
\newcommand{\vv}{\boldsymbol{v}}
\newcommand{\vhat}{\hat{\mathbf{v}}}
\newcommand{\av}{\boldsymbol{a}}
\newcommand{\aperp}{\boldsymbol{a}_{\perp}}
\newcommand{\aph}{\hat{\mathbf{a}}_{\perp}}
\newcommand{\Bh}{\hat{\mathbf{b}}}
\newcommand{\Av}{\boldsymbol{A}}
\newcommand{\sig}{\boldsymbol{\sigma}}
\newcommand{\ee}{\hat{\mathbf{e}}}
\newcommand{\dd}{\mathrm{d}}

\DeclareMathOperator{\atantwo}{atan2}

\newcommand{\PaperAbstract}{%
When the control field of a qubit, a polarization state, or a spin-$\tfrac12$ system is
swept through a level degeneracy, its direction traces an open curve on the Bloch sphere
whose endpoints are antipodal, and the geodesic rule for the open-path geometric phase
becomes ambiguous: infinitely many geodesics close the path, and different closures
enclose different solid angles. We resolve this ambiguity in closed form. A
coordinate-free monopole connection defines the open-path solid angle $\Omega[C]$
intrinsically, and displacing the degeneracy by $\epsilon\uhat$ closes the path with
enclosed solid angle $\Omega(\epsilon\uhat)=\Omega[C]+2\alpha+O(\epsilon)$, where $\alpha$
is the azimuth of the transverse part of $\uhat$ measured from the principal normal of the
control curve at the crossing. The identity between geometric phase and enclosed solid
angle therefore holds for exactly one closing geodesic---the great circle in the
osculating plane ($\alpha=0$)---supplied by the curvature at the degeneracy. Berry's $\pi$
invariant under reversal of the displacement and the values $\pm\pi/2$ under a reflection
symmetry follow as corollaries, and the pure-state limit of the finite-temperature Uhlmann
phase selects the osculating-plane closure automatically, turning the heuristic closing
rules of the open-path literature into a computable prescription.}
\newcommand{\PaperKeywords}{geometric phase, Berry phase, solid angle, open paths, geodesic rule, monopole connection, Uhlmann phase}

\begin{document}

\ifdef{\ClassIOPJournal}{%
  \articletype{Paper}
  \title{Geometric phase of open paths and a geodesic-selection rule at a level degeneracy}
  \author{Hyeonseok Yang$^{1,*}$ and Changsuk Noh$^1$}
  \affil{$^1$Department of Physics, Kyungpook National University, Daegu, South Korea}
  \affil{$^*$Author to whom any correspondence should be addressed.}
  \email{yhs1802@gmail.com}
  \keywords{\PaperKeywords}
  \begin{abstract}\PaperAbstract\end{abstract}
}{%
  \title{Geometric phase of open paths and a geodesic-selection rule at a level degeneracy}
  \author{Hyeonseok Yang\thanks{Corresponding author: yhs1802@gmail.com}~~and Changsuk Noh\\
          \small Department of Physics, Kyungpook National University, Daegu, South Korea}
  \date{\today}
  \maketitle
  \begin{abstract}\PaperAbstract\end{abstract}
  \vspace{1pc}\noindent{\it Keywords}: \PaperKeywords\par\vspace{1pc}
}

\section{Introduction}
\label{sec:intro}

For a closed curve $C$ on the Bloch (or Poincar\'e) sphere $S^{2}$, the relation between
geometric phase and solid angle is well established:
a spin-$\tfrac12$ system acquires the Berry
phase $\gamma_{g}[C]=-\tfrac12\Omega[C]$, where $\Omega[C]$ is the signed spherical area enclosed by $C$~\cite{Berry1984,Simon1983,Pancharatnam1956}. When the curve is instead
\emph{open}, the kinematic formulation defines its geometric phase by completing the curve with a geodesic arc connecting the endpoints~\cite{AharonovAnandan1987,SamuelBhandari1988,MukundaSimon1993}. Because parallel transport along this geodesic contributes no geometric phase, the open-path geometric phase equals minus one-half of the solid angle enclosed by the resulting closed curve. This is known as the \emph{geodesic rule}.

This construction, however, becomes ambiguous when the endpoints of the open curve are antipodal. Such a situation can arise naturally in systems with an effective two-level description, including a driven qubit, a spin-$\tfrac{1}{2}$ in a rotating field, and a light beam passing through a polarization singularity. Consider a two-level Hamiltonian
\begin{equation}
H(t)=-\mathbf{d}(t)\cdot\boldsymbol{\sigma},
\label{eqn:Hamiltonian}
\end{equation}
where $\boldsymbol{\sigma}=(\sigma_x,\sigma_y,\sigma_z)$ is the vector
of Pauli matrices, and $\dvec(t)\in\mathbb{R}^3$ is a smooth,
$T$-periodic control vector, $\dvec(t+T)=\dvec(t)$. Thus, the path traced by $\dvec(t)$ is a closed curve in the control-parameter space. We assume that this curve passes through the origin at a single
isolated point $t_0$ in each period
and that the velocity $\vv\equiv\dot{\mathbf{d}}(t_0)$ is nonzero. We refer to such a passage through the origin as a transversal crossing. At $t_0$, the two eigenvalues $\mp|\mathbf{d}(t)|$ coincide, and the Hamiltonian becomes degenerate. Let $\hat{\mathbf v}\equiv\vv/|\vv|$ denote the direction
in which the curve crosses the origin. 
In the vicinity of the crossing, $\dvec(t)\simeq\vv\,(t-t_{0})$
and therefore
\begin{equation}
\lim_{t\to t_0^\pm}\hat{\mathbf{n}}(t)=\pm\hat{\mathbf{v}},
\end{equation}
where $\hat{\mathbf{n}}(t)\equiv\mathbf{d}(t)/|\mathbf{d}(t)|$. Thus, as $t$ passes through $t_0$, $\hat{\mathbf{n}}(t)$ jumps discontinuously between two antipodal points on $S^2$. 
Since $\nhat(t)$ is undefined at $t_{0}$, its trajectory on $S^{2}$ is not a closed loop but an open curve $C$, with the antipodal endpoints $\pm\vhat$.
Traversed over one period, $C$ begins at $+\vhat$, immediately after the crossing, and ends at $-\vhat$, immediately before it. Between such antipodal endpoints the geodesic rule no longer singles out a closure: just as all meridians join the north and south poles, infinitely many great-circle arcs of equal length join $+\vhat$ to $-\vhat$.

At first sight, this ambiguity may appear harmless: a closing geodesic is commonly said to contribute no geometric phase. This argument is well defined for non-antipodal endpoints, which determine a unique shorter geodesic and correspond to non-orthogonal states with a well-defined Pancharatnam phase~\cite{Pancharatnam1956,SamuelBhandari1988}. For antipodal endpoints, however, the corresponding states are orthogonal, their relative phase is undefined, and infinitely many half-great-circle arcs are available. Different closures therefore produce closed loops with different solid angles and geometric phases. Although the relation $\gamma_g=-\tfrac{1}{2}\Omega$ remains valid for each closed loop, the geometric phase assigned to the original open path is not unique without an additional criterion. Selecting the appropriate closing geodesic thus carries physical content.

This ``which geodesic'' problem, together with the $\pm\pi$ phase jumps that occur when the evolving state becomes orthogonal to the initial one---geometrically, when the trajectory passes through the point antipodal to its starting point on $S^{2}$---has been recognized previously.
Rakhecha and Wagh showed for a two-level system that, as the path passes through this antipodal point, the geometric-phase jump is determined by the azimuthal separation between the two geodesic branches, yielding a jump of $\pm\pi$ for a smooth crossing~\cite{RakhechaWagh1996}.
More recently, a hierarchy of operational rules for the open-path Pancharatnam geometric phase was proposed, culminating in a prescription that selects the closing geodesic according to the direction from which the physical path approaches the endpoint~\cite{GarzaSotoHagen2023}. Predictions of the geodesic rule, including the associated sign changes and $\pi$ phase jumps, have also been verified experimentally~\cite{Zhou2020,Wagh1998}. These studies establish the operational prescription and its observable consequences, but a general coordinate-free derivation remains lacking.

The aim of this work is to develop a closed-form, coordinate-free treatment of the geodesic-selection problem, identify the mechanism underlying its dependence on the chosen closure, and formulate an explicit selection rule.
Our analysis is based on the coordinate-free monopole one-form connection,
\begin{align}
\mathcal{A} \equiv \Av\cdot\dd\nhat=\frac{\nz\cdot\bigl(\nhat\times\dd\nhat\bigr)}{1+\nz\cdot\nhat},
\end{align}
which allows us to obtain the following results:
\begin{enumerate}
\item[(i)] define the solid angle associated with an open curve $C$ having antipodal endpoints as a line integral over the curve alone, with its starting point taken as the base point, and define the corresponding open-path geometric phase as $\gamma_g[C]=-\tfrac12\Omega[C]$;

\item[(ii)] show that, among the infinitely many geodesic closures of the antipodal endpoints, a unique oriented geodesic arc satisfies the geodesic rule for the open-path geometric phase defined in (i), and identify it as the closure selected by the osculating plane of $\mathbf{d}(t)$ at the degeneracy;

\item[(iii)] introduce a regularization
$\mathbf{d}_{\epsilon}(t)=\mathbf{d}(t)+\epsilon\hat{\mathbf{u}}$
and show that the solid angle of the resulting nonsingular closed path satisfies $\Omega(\epsilon\hat{\mathbf{u}}) = \Omega[C]+2\alpha+O(\epsilon).$ Here, $\alpha$ is a signed angle between two directions transverse to the crossing: the
direction in which the shift $\epsilon\uhat$ carries the curve past the degeneracy, and the direction in which the curve itself is bending there;

\item[(iv)] recover the universal $\pi$ phase difference under reversal of the regularization direction and the values $\pm\pi/2$ under reflection symmetry;

\item[(v)] show that the pure-state limit of the finite-temperature Uhlmann
phase~\cite{Uhlmann1986,Viyuela2014,Andersson2016,Wang2023SciPost,Wang2025PRB}  automatically selects the
osculating-plane regularization $\alpha=0$.
\end{enumerate}

The deviation $2\alpha$ originates from the unequal contributions accumulated near the two antipodal endpoints as the regularized path avoids the degeneracy. To recover the open-path geometric phase defined above, the path should be closed using the oriented geodesic selected by the osculating plane, which is determined by the local curvature of $\mathbf{d}(t)$ at the degeneracy. Our analysis thus provides both a geometric explanation and an operational prescription. The rest of the manuscript is organized as follows. Sections~2--4 develop these results, Section~5 concludes, and Appendix~A presents numerical verification.

\section{The monopole connection on $S^{2}$}
\label{sec:monopole}

\subsection{Determining the connection}
\label{sec:conn}

To construct a connection suitable for treating open paths, we choose
a base point $\hat{\mathbf n}_{0}\in S^{2}$ and seek a one-form
$\mathcal{A} =\mathbf A\cdot d\hat{\mathbf n}$ whose circulation around a closed
curve equals the enclosed solid angle.
Since $\nhat$ is a unit vector,
$\dd\nhat$ lies in the tangent plane of $S^{2}$ and satisfies $\nhat\cdot\dd\nhat=0$.
Consequently, only the tangential component of $\Av$ contributes to $\mathcal{A}$:
adding an arbitrary radial term $h(\nhat)\,\nhat$ leaves the one-form unchanged, so we may
choose $\Av\cdot\nhat=0$ without loss of generality. Two tangential directions remain, the
polar (meridional) direction $\ee_{\theta}$ and the azimuthal one $\ee_{\phi}$, and we now
show that only the latter carries physical content.

Consider first the component of $\Av$ along $\ee_{\theta}$. Rotational symmetry about the $\nz$ axis restricts this component to the form $g(\theta)\,\ee_{\theta}$. Since
$\ee_{\theta}\cdot\dd\nhat=\dd\theta$, its contribution to the one-form is
$g(\theta)\,\dd\theta$, which is pure gauge: it leaves the curvature unchanged, since  $\dd[g(\theta)\,\dd\theta]=0$. Setting $g\equiv 0$ is thus a choice of gauge rather than a restriction. We adopt it here because it makes the great circles through $\pm\nz$
parallel-transport paths: on such a circle $\dd\nhat$ points along $\ee_{\theta}$, so with  $g\equiv0$ a geodesic joining any point $\nhat$ to the base point $\nz$ contributes nothing to the line integral. This property is established in Lemma~\ref{lem:geodesic} and underlies the geodesic-closure construction for open paths developed below.

The remaining tangential component lies along the azimuthal direction $\ee_{\phi}$ , which is parallel to $\nz\times\nhat$. Unlike $g$, its coefficient does affect the curvature and cannot be chosen freely; it is fixed uniquely by Stokes' theorem below. Rotational symmetry about the $\nz$
axis further requires this coefficient to depend only on the polar angle $\theta$, defined by
$\cos\theta\equiv\nz\cdot\nhat$. We therefore write
\begin{equation}
\Av=f(\theta)\,(\nz\times\nhat),\qquad \cos\theta\equiv\nz\cdot\nhat .
\label{eq:ansatz}
\end{equation}

To determine $f$, take $C_{\theta}$ to be the circle at constant polar angle $\theta$ about the 
 $\hat{\mathbf{n}}_0$ axis. The spherical cap it bounds has the solid angle
\begin{equation}
\Omega[C_{\theta}]=\int_{0}^{2\pi}\!\!\int_{0}^{\theta}\sin\theta'\,\dd\theta'\,\dd\phi
=2\pi(1-\cos\theta).
\label{eq:cap}
\end{equation}
On $C_{\theta}$ the magnitude $|\nz\times\nhat|=\sin\theta$ is constant, $\Av$ is parallel
to $\dd\nhat$, and the circumference is $2\pi\sin\theta$, so
\begin{equation}
\oint_{C_{\theta}}\Av\cdot\dd\nhat
=\big(f(\theta)\sin\theta\big)\big(2\pi\sin\theta\big)=2\pi f(\theta)\sin^{2}\theta .
\label{eq:circ}
\end{equation}
Equating~(\ref{eq:cap}) and~(\ref{eq:circ}) through Stokes' theorem gives
$f(\theta)=(1-\cos\theta)/\sin^{2}\theta=1/(1+\cos\theta)$, and restoring
$\cos\theta=\nz\cdot\nhat$ yields the coordinate-free form,
\begin{equation}
\Av(\nhat)=\frac{\nz\times\nhat}{1+\nz\cdot\nhat}.
\label{eq:A}
\end{equation}

\subsection{The one-form connection  and its curvature}
\label{sec:oneform}

Introduce spherical coordinates $(\theta,\phi)$ with $\nz$ at the north pole, i.e., $\nz=\hat{\mathbf z}$, so that
$\nhat=(\sin\theta\cos\phi,\sin\theta\sin\phi,\cos\theta)$. With
the right-handed orthonormal frame $\{\nhat,\ee_{\theta},\ee_{\phi}\}$ one has
$\nhat\times\ee_{\theta}=\ee_{\phi}$, $\nhat\times\ee_{\phi}=-\ee_{\theta}$, and
$\dd\nhat=\ee_{\theta}\,\dd\theta+\sin\theta\,\ee_{\phi}\,\dd\phi$. A short computation using
$\hat{\mathbf z}\cdot\ee_{\phi}=0$ and $\hat{\mathbf z}\cdot\ee_{\theta}=-\sin\theta$ gives the one-form \emph{connection}
\begin{equation}
\mathcal{A}=\frac{\nz\cdot(\nhat\times\dd\nhat)}{1+\nz\cdot\nhat}
=\frac{\sin^{2}\theta}{1+\cos\theta}\,\dd\phi=(1-\cos\theta)\,\dd\phi .
\label{eq:oneform}
\end{equation}
The exterior derivative of this is the coordinate-free two-form \emph{curvature},
$\dd\big[(1-\cos\theta)\,\dd\phi\big]=\sin\theta\,\dd\theta\wedge\dd\phi\equiv\dd\Omega$, so
for a closed curve $C=\partial\Sigma$,
\begin{equation}
\oint_{C}\Av\cdot\dd\nhat=\iint_{\Sigma}\dd\Omega=\Omega[C].
\label{eq:stokes}
\end{equation}

Before proceeding we fix the relation to the familiar Berry phase. For the two-level
Hamiltonian in equation (\ref{eqn:Hamiltonian}), we follow the eigenstate $|\psi_{-}(\nhat)\rangle$ of energy $-|\dvec|$ aligned with $\nhat$. A standard computation of its Berry connection in
the gauge with base point $\nz$ gives the familiar monopole form
$a_{B}=-\tfrac12(1-\cos\theta)\,\dd\phi$, which by equation~(\ref{eq:oneform}) is nothing but
\begin{equation}
a_{B}=-\tfrac12\,\mathcal{A} .
\label{eq:gg}
\end{equation}
The Berry connection and the one-form studied here therefore differ only by the constant factor $-\tfrac{1}{2}$, and we work with $\mathcal{A}$ in this paper.

\subsection{Gauge, base point and Dirac string}
\label{sec:gauge}

The curvature derived above is that of a monopole located at
$\mathbf{d}=0$. A gauge potential for a monopole cannot be chosen
smoothly over the entire sphere. In the gauge defined by
$\hat{\mathbf{n}}_0$, the connection~(\ref{eq:A})
is regular everywhere except at
$\hat{\mathbf{n}}=-\hat{\mathbf{n}}_0$, where its denominator vanishes.
This point is the intersection of the unit sphere with the Dirac
string~\cite{Dirac1931}, an unphysical line singularity of the monopole gauge potential
in the full parameter space. Changing $\hat{\mathbf{n}}_0$ moves this
singularity and therefore corresponds to a change of gauge.
For a closed
curve, $\Omega$ is independent of $\nz$ by equation~(\ref{eq:stokes}) (exactly so if the curve does
not cross the string, and up to an integer multiple of $4\pi$ otherwise). For an
\emph{open} curve, however, a change of base point acts as a gauge transformation
$a_{B}\to a_{B}+\dd\chi$ that leaves a boundary term
$\int_{C}\dd\chi=\chi(\text{end})-\chi(\text{start})\neq0$. The choice of base point is thus
an additional gauge---a path-dependent term---for open curves, and the value of $\Omega[C]$
depends on it. It is therefore natural to tie the base point to the curve itself. In the
situation of interest this choice is essentially forced: taking $\nz=+\vhat$ (the starting
direction) places the string at the opposite endpoint $-\vhat$, so that $C$ terminates
exactly on the string. As we show in section~\ref{sec:openpath}, this is precisely what
renders the non-uniqueness of the closing geodesic harmless, and we adopt
$\nz=\pm\vhat$ henceforth (Definition~\ref{def:omega}).


\section{Solid angle of an open path with antipodal endpoints}
\label{sec:openpath}

\subsection{Geodesic rule and intrinsic definition}
\label{sec:georule}

\begin{lemma}[coordinate-free geodesic rule]
\label{lem:geodesic}
Along the geodesic (great-circle) arc from any $\nhat\in S^{2}$ to the base point $\nz$ one has $\Av\cdot\dd\nhat=0$ identically; consequently a geodesic arc contributes nothing to the
geometric phase.
\end{lemma}

\begin{proof}
The geodesic lies in the plane $\Pi_{0}\equiv\operatorname{span}\{\nz,\nhat\}$. Any point
$\nhat'$ on it, and its tangent $\dd\nhat$, lie in $\Pi_{0}$, whereas
$\Av\propto\nz\times\nhat'$ is by construction orthogonal to $\Pi_{0}$. The inner product of a vector normal to a plane with one lying in it vanishes, so $\Av\cdot\dd\nhat=0$.
\end{proof}

\begin{definition}[open-path solid angle]
\label{def:omega}
For an open curve $C\subset S^{2}$ starting at $\nz$, define
\begin{equation}
\Omega[C]\equiv\int_{C}\Av\cdot\dd\nhat
=\int_{C}\frac{\nz\cdot(\nhat\times\dd\nhat)}{1+\nz\cdot\nhat},
\label{eq:defomega}
\end{equation}
with the base point taken to be the starting point of the curve.
\end{definition}

The line integral~(\ref{eq:defomega}) admits a purely geometric reading. Discretising $C$
as $\{\nhat^{(0)}=\nz,\allowbreak\nhat^{(1)},\allowbreak\dots,\allowbreak\nhat^{(M)}\}$ and letting $\Delta_{i}$ be the signed
solid angle, seen from the centre of the sphere, of the spherical triangle
$(\nz,\nhat^{(i)},\nhat^{(i+1)})$, one has $\Omega[C]=\lim_{\max|\dd\nhat|\to0}\sum_{i}\Delta_{i}$.
Indeed, the signed solid angle of a spherical triangle is given by the van
Oosterom--Strackee formula~\cite{VanOosteromStrackee1983},
\begin{equation}
\Delta(\boldsymbol \alpha,\boldsymbol \beta,\boldsymbol \gamma)
=2\arctan\frac{\boldsymbol \alpha\cdot(\boldsymbol \beta\times\boldsymbol \gamma)}
{1+\boldsymbol \alpha\cdot\boldsymbol \beta+\boldsymbol \beta\cdot\boldsymbol \gamma+\boldsymbol \gamma\cdot\boldsymbol \alpha},
\label{eq:vos}
\end{equation}
and inserting $\boldsymbol \alpha=\nz$, $\boldsymbol \beta=\nhat$, $\boldsymbol \gamma=\nhat+\dd\nhat$ and
expanding to first order reproduces the integrand of equation~(\ref{eq:defomega}). The
``solid angle seen from the centre'', the ``base-point triangulation'', and the ``one-form
line integral'' therefore coincide.

\subsection{Independence of the closing geodesic}
\label{sec:antipodal}

\begin{lemma}[antipodal endpoints]
\label{lem:antipodal}
If the endpoint of $C$ is the antipode $-\nz$ of the base point, then any geodesic joining $-\nz$ to $\nz$ contributes zero to the line integral in equation~(\ref{eq:defomega}). Hence the line
integral $\Omega[C]$ is determined by the curve $C$ alone, independently of the closing
path.
\end{lemma}

\begin{proof}
A great circle through $\pm\nz$ is the intersection of the sphere with a plane
$\operatorname{span}\{\nz,\hat{\mathbf m}\}$ for some unit $\hat{\mathbf m}\perp\nz$. This
plane contains $\nz$, so Lemma~\ref{lem:geodesic} applies at each of its points; since
$\hat{\mathbf m}$ was arbitrary, the claim follows.
\end{proof}

\begin{remark}
\label{rem:notmean}
Lemma~\ref{lem:antipodal} states that the \emph{line integral} $\Omega[C]$ is intrinsic to
$C$; it does not state that the solid angle \emph{enclosed} after closing is independent of
the chosen geodesic. Two different closings bound a spherical lune between them, whose area
on the unit sphere equals twice the angle $\alpha$ at which the two great circles meet at the shared endpoints. Consequently the identity ``geometric phase $=$ enclosed solid angle'' holds for exactly one closing geodesic, and $\alpha$---defined as the azimuth of the chosen closing relative to that canonical one---measures the failure of the identity for any other choice. Identifying the canonical geodesic (the osculating-plane great circle, $\alpha=0$) and proving the deviation $2\alpha$ is the subject of section~\ref{sec:main}, and gives a closed form to the operational selection rule of~\cite{RakhechaWagh1996,GarzaSotoHagen2023}. Finally, the base point is placed at the
start of $C$ so that no closing geodesic passes through $\nz$ and the line integral carries
no spurious endpoint contribution.
\end{remark}

\subsection{Osculating plane and Frenet frame}
\label{sec:frenet}

The point $\dvec(t_{0})=\mathbf0$ (the degeneracy) is, as noted in section~\ref{sec:intro},
where $\nhat$ jumps between the antipodes $\mp\vhat$; the ``degeneracy'' and the ``antipodal
endpoints of the $\nhat$-trajectory'' are the same event. Set $\dvec(t_{0})=\mathbf0$,
$\vv\equiv\dot\dvec(t_{0})\neq\mathbf0$, $\av\equiv\ddot\dvec(t_{0})$ and $s\equiv t-t_{0}$.
Then $\dvec(t_{0}+s)=\vv s+\tfrac12\av s^{2}+O(s^{3})$,
\begin{equation}
\nhat(s)=\operatorname{sgn}(s)\Big[\vhat+\frac{s}{2|\vv|}\aperp\Big]+O(s^{2}),
\label{eq:nlocal}
\end{equation}
where
\begin{equation}
\aperp\equiv\av-a_{\parallel}\,\vhat,\qquad a_{\parallel}\equiv\av\cdot\vhat,
\end{equation}
so that $\nhat\to\pm\vhat$ as $s\to0^{\pm}$: the direction field jumps between antipodes.

\begin{figure}
    \centering
    \includegraphics[width=0.85\linewidth]{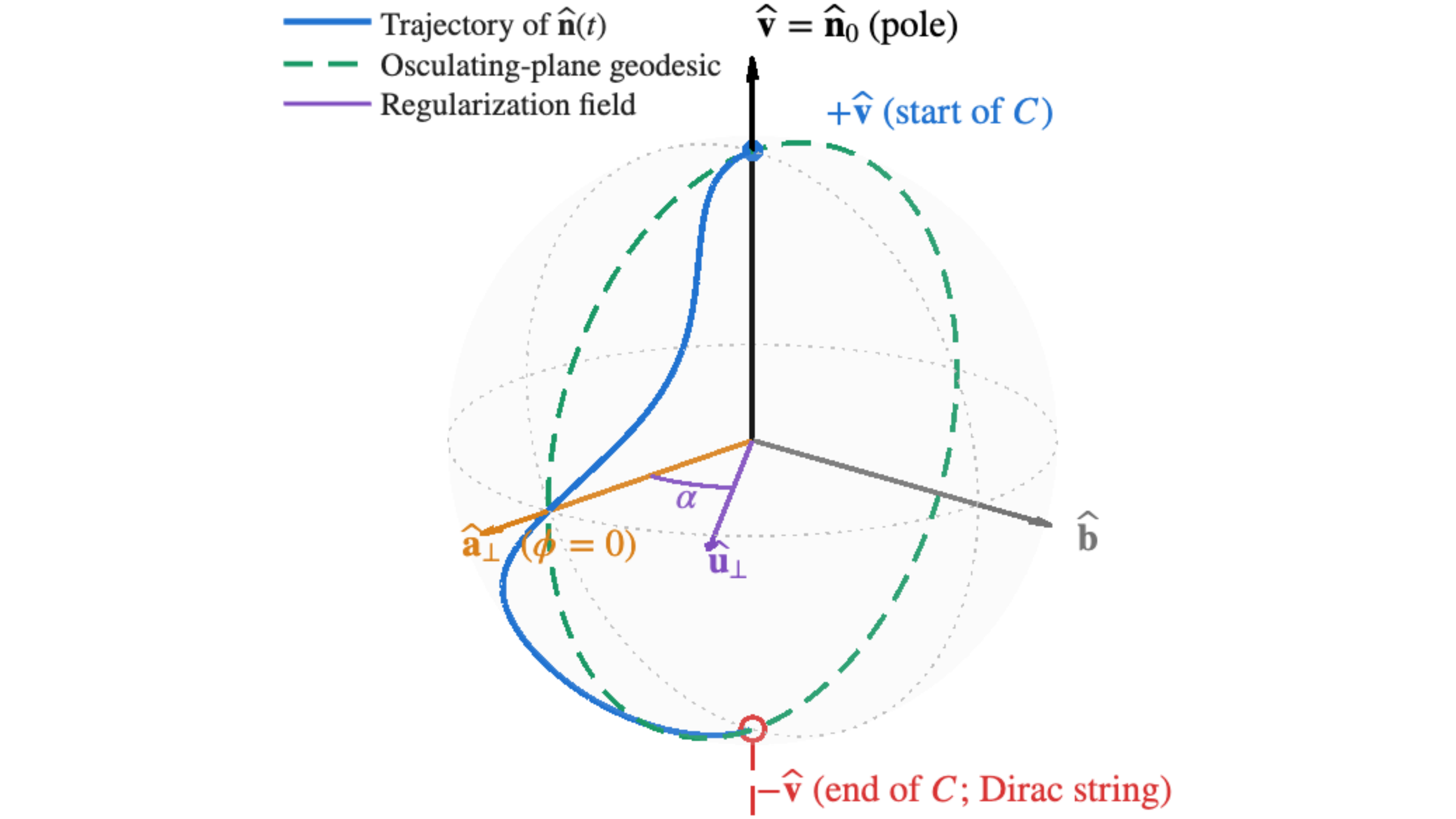}
    \caption{Shown for the example field $\dvec(t)=\bigl[1-\cos t,\ \tfrac12\sin t\,(1-\cos t),\ \sin t\bigr]$ with degeneracy at $t_{0}=0$. The direction field $\nhat=\dvec/|\dvec|$ traces an open curve $C$ (solid blue): the endpoints are the antipodal pair $\pm\vhat$, with base point $\nz=+\vhat$ and the Dirac string at $-\vhat$ (dashed red). The Frenet frame $\{\vhat,\aph,\Bh\}$ sets pole and azimuthal origin; the osculating-plane geodesic (dashed green) is tangent to $C$ at both endpoints ($\phi=0$), and a regularization $\uhat_{\perp}$ (purple) deviates from it by the azimuth $\alpha$.}
    \label{fig:frenet}
\end{figure}

Assume $\aperp\neq\mathbf0$ (equivalently $\av \nparallel \vv$), i.e., that the curve $\mathbf{d}(t)$ has nonzero curvature at $t_0$.
Introduce the Frenet frame at $t_{0}$, with orthonormal basis consisting of the tangent $\vhat$, principal normal $\aph=\aperp/|\aperp|$, and binormal $\Bh=\vhat\times\aph$, and let $\Pi_{\rm osc}\equiv\operatorname{span}\{\vhat,\aph\}
=\operatorname{span}\{\vv,\av\}$ be the osculating plane.
The geometry is illustrated in figure~\ref{fig:frenet}. Since the $\Bh$-component of $\dvec$ is $O(s^{3})$, the curve lies
in $\Pi_{\rm osc}$ to second order. Decomposing $\dvec$ in the Frenet frame, $\dvec(s)=\xi(s)\,\vhat+\eta(s)\,\aph+\zeta(s)\,\Bh$,
with components
\begin{equation}
\begin{aligned}
\xi &\equiv \dvec\cdot\vhat \;=\; |\vv|\,s+\tfrac12 a_{\parallel}s^{2}+O(s^{3}),\\
\eta &\equiv \dvec\cdot\aph \;=\; \tfrac12|\aperp|\,s^{2}+O(s^{3}),\\
\zeta &\equiv \dvec\cdot\Bh \;=\; O(s^{3}).
\end{aligned}
\label{eq:split}
\end{equation}
To leading order the curve is the parabola
\begin{equation}
\eta=\frac{|\aperp|}{2|\vv|^{2}}\,\xi^{2},
\label{eq:parabola}
\end{equation}
with vertex at the origin, opening towards $+\aph$.

We now fix the angular variables used below. The pole of the spherical coordinates has already been tied to the curve, $\nz=\vhat$; we further choose the azimuthal origin in the transverse plane
spanned by $(\aph,\Bh)$ so that $\phi=0$ along $+\aph$ and $\phi=\pi/2$ along $+\Bh$. The
transverse part of $\nhat$---its projection onto this plane,
$\nhat-(\nhat\cdot\vhat)\,\vhat=(\eta\,\aph+\zeta\,\Bh)/|\dvec|$ with $\eta,\zeta$ as in
equation~(\ref{eq:split})---then has azimuth $\phi=\atantwo(\zeta,\eta)$, which is what we
mean by the \emph{transverse azimuth} of $\nhat$. We call the trajectory of the
unregularized direction field $\nhat(s)$ of equation~(\ref{eq:nlocal}) the \emph{bare}
curve, in contradistinction to the regularized field introduced later, and write $\phi_{\rm bare}$ for its azimuth.

\begin{lemma}[harmlessness of the degeneracy]
\label{lem:bare}
Near the degeneracy the transverse azimuth of the bare curve is
\begin{equation}
\phi_{\rm bare}(s)=\frac{b_{3}}{3\,|\aperp|}\,s+O(s^{2}),
\qquad b_{3}\equiv\dddot{\dvec}(t_{0})\cdot\Bh ,
\label{eq:barephi}
\end{equation}
Regardless of $b_{3}$, the contribution of any
neighbourhood $|s|<w$ of the degeneracy to $\Omega[C]$ is $O(w)$ and hence vanishes as
$w\to0$: the degeneracy contributes nothing to the intrinsic solid angle.
\end{lemma}

\begin{proof}
Carrying the expansion in equation~(\ref{eq:split}) one order further, the binormal component is
$\zeta(s)=\tfrac16 b_{3}s^{3}+O(s^{4})$, while $\eta(s)=\tfrac12|\aperp|s^{2}+O(s^{3})>0$ on
both branches. Hence
$\phi_{\rm bare}=\atantwo(\zeta,\eta)=\zeta/\eta+\cdots=\frac{b_{3}}{3|\aperp|}\,s+O(s^{2})$,
an odd function that passes through zero with finite slope rather than remaining pinned at
$0$. For the contribution to $\Omega[C]$ through the one-form in equation~(\ref{eq:oneform}), split the
window $|s|<w$ by branch. On the branch $s>0$ the polar angle is
$\theta\simeq(|\aperp|/2|\vv|)\,s\to0$, so the weight is $1-\cos\theta=O(s^{2})$; with
$\dd\phi_{\rm bare}=O(1)\,\dd s$ this branch contributes $O(w^{3})$. On the branch $s<0$ the
polar angle approaches $\pi$, so the weight is $2-O(s^{2})$, and the branch contributes
$2\big[\phi_{\rm bare}(0)-\phi_{\rm bare}(-w)\big]+O(w^{3})
=\frac{2b_{3}}{3|\aperp|}\,w+O(w^{2})=O(w)$. The total contribution goes to $0$ as the size of the window tends to zero, $w\to0$.
\end{proof}

The coefficient $b_{3}$ is proportional to the Frenet torsion of $\dvec$ at the crossing,
so the azimuth is frozen at $0$ precisely for the planar (reflection-symmetric) case of
Corollary~\ref{cor:reflection}.
The parabola in equation~(\ref{eq:parabola}) and Lemma~\ref{lem:bare} together carry the two facts on
which the rest of the paper turns. First, both branches ($s\gtrless0$) of the curve lie on the $+\aph$ side, since the coefficient $\tfrac12 s^{2}$ in $\eta$ is positive irrespective
of the sign of $s$; the transverse direction of $\nhat$ therefore agrees with $+\aph$ to
leading order on both sides of the degeneracy, deviating only at $O(s)$ through the
torsion, and by Lemma~\ref{lem:bare} the jump of $\nhat$ from $-\vhat$ to $+\vhat$ leaves
no concentrated contribution to the solid angle. Second, the tangent directions of the
parabola at $\pm\vhat$ single out, among the infinitely many geodesics joining the
antipodes, the one great circle actually tangent to the curve---the osculating-plane great
circle; this tangency is a second-order statement, unaffected by the torsion. The identity
of the ``geodesic that realizes the enclosed solid angle'', anticipated in
Remark~\ref{rem:notmean}, is thus supplied by the curvature at the degeneracy.

\section{Dependence on the regularization direction}
\label{sec:main}
The open path $C$ terminates at the antipodal points $\pm\vhat$, where the geodesic that closes it is not unique. One way to resolve this is by displacing the degeneracy off the
origin: a small shift $\epsilon\uhat$ makes $\dvec_{\epsilon}$ nonvanishing, so $\nhat_{\epsilon}$
is a genuine closed loop with an unambiguous solid angle, and letting $\epsilon\to0$ recovers a
definite closing whose selection depends on the direction $\uhat$.

\subsection{The regularized field}
\label{sec:reg}

Let us introduce the regularization $\dvec_{\epsilon}\equiv\dvec+\epsilon\uhat$ with $\epsilon>0$
and $\uhat\in S^{2}$ having a nonzero transverse component, i.e., $\uhat\nparallel\vhat$. Then $\dvec_{\epsilon}(t_{0})=\epsilon\uhat\neq\mathbf0$ avoids the
origin, and $\nhat_{\epsilon}$ is a smooth closed curve. Decomposing as in equation~(\ref{eq:split}),
\begin{equation}
\begin{aligned}
\xi_{\epsilon}(s)&=|\vv|s+\tfrac12 a_{\parallel}s^{2}+\epsilon(\uhat\cdot\vhat)+O(s^{3}),\\
\eta_{\epsilon}(s)&=\tfrac12|\aperp|s^{2}+\epsilon(\uhat_{\perp}\cdot\aph)+O(s^{3}),\\
\zeta_{\epsilon}(s)&=\epsilon(\uhat_{\perp}\cdot\Bh)+O(s^{3}),
\end{aligned}
\label{eq:regsplit}
\end{equation}
where $\uhat_{\perp}\equiv\uhat-(\uhat\cdot\vhat)\vhat$.
The transverse plane is spanned by $(\aph,\Bh)$ with coordinates $(\eta_{\epsilon},\zeta_{\epsilon})$. Writing
\begin{equation}
u\equiv|\uhat_{\perp}|,\qquad
\alpha\equiv\atantwo\big(\uhat_{\perp}\cdot\Bh,\;\uhat_{\perp}\cdot\aph\big),
\label{eq:alpha}
\end{equation}
so that $\uhat_{\perp}=u(\cos\alpha\,\aph+\sin\alpha\,\Bh)$, the transverse components become
\begin{equation}
\eta_{\epsilon}(s)=\tfrac12|\aperp|s^{2}+\epsilon u\cos\alpha,\qquad
\zeta_{\epsilon}(s)=\epsilon u\sin\alpha .
\label{eq:etazeta}
\end{equation}
Geometrically, $\alpha$ is the angle specifying the direction of the transverse
regularization $\uhat_{\perp}$ within the $(\aph,\Bh)$ plane, measured from the
principal normal $\aph$; as we now show, it also parametrizes the closing geodesic:
$\alpha=0$ corresponds to the osculating-plane great circle (the canonical geodesic),
and $\alpha\neq0$ to a great circle rotated by $\alpha$ about the $\vhat$ axis (see figure~\ref{fig:frenet}).

The key structural fact is that the regularization contributes an
\emph{$s$-independent} transverse shift $\epsilon\uhat_{\perp}$. To leading order the transverse trajectory is then the horizontal line
$\zeta\simeq\epsilon u\sin\alpha$ traced as $s$ runs through the crossing: the point
$(\eta_{\epsilon},\zeta_{\epsilon})$ comes in from large $\eta$, reaches its closest approach to
the origin at $s=0$, where $\eta_{\epsilon}=\epsilon u\cos\alpha$, and recedes back to large
$\eta$, staying at fixed height $\zeta\simeq\epsilon u\sin\alpha$ throughout. The azimuth is accordingly

\begin{equation}
\phi_{\epsilon}(s)=\atantwo\big(\zeta_{\epsilon},\eta_{\epsilon}\big)
=\atantwo\Big(\epsilon u\sin\alpha,\;\tfrac12|\aperp|s^{2}+\epsilon u\cos\alpha\Big).
\label{eq:phieps}
\end{equation}

According to equation~\eqref{eq:oneform}, an azimuthal change $\dd\phi$ contributes $(1-\cos\theta)\dd\phi$ to the solid angle. Its contribution therefore depends strongly on the polar angle: the weight vanishes at the north pole $(\theta=0)$ and reaches its maximum value of $2$ at the south pole $(\theta=\pi)$. To determine the net contribution of the azimuthal turn in equation~\eqref{eq:phieps}, we compare the characteristic $s$-scales over which $\theta_\epsilon$ and $\phi_\epsilon$ vary.
The transverse vector entering equation~\eqref{eq:phieps} can be written as
\begin{equation}
(\eta_\epsilon,\zeta_\epsilon)
=
\tfrac12|\aperp|s^{2}(1,0)
+
\epsilon u(\cos\alpha,\sin\alpha),
\end{equation}
which shows that the quadratic contribution becomes comparable to the magnitude of the
regulariza\-tion-induced transverse displacement at
\begin{equation}
|s|\sim s_\phi
\equiv
\sqrt{\frac{2\epsilon u}{|\aperp|}}.
\end{equation}

The polar angle is governed by a different scale, which is best seen through its deviation
from the nearer pole, $\delta\equiv\min(\theta_{\epsilon},\pi-\theta_{\epsilon})$. From $\tan\theta_\epsilon=\sqrt{\eta_{\epsilon}^{2}+\zeta_{\epsilon}^{2}}\,/\,|\xi_{\epsilon}|$ and $\tan\delta\simeq\delta$ for small $\delta$, $|\xi_{\epsilon}|\simeq|\vv||s|$ to leading
order, and the transverse part dominated by the regularization for $|s|\ll s_{\phi}$,
$\sqrt{\eta_{\epsilon}^{2}+\zeta_{\epsilon}^{2}}\simeq\epsilon u$, the deviation obeys
\begin{equation}
\delta(s)\simeq\frac{\sqrt{\eta_{\epsilon}^{2}+\zeta_{\epsilon}^{2}}}{|\xi_{\epsilon}|}
\simeq\frac{\epsilon u}{|\vv|\,|s|}\,.
\label{eq:delta}
\end{equation}
Equation~(\ref{eq:delta}) yields the polar scale: the deviation is small precisely while
$|s|\gg s_{\theta}\equiv\epsilon u/|\vv|$, so that the polar angle stays pinned at its
pole, $\theta_{\epsilon}\approx0$ for $s>0$ and $\theta_{\epsilon}\approx\pi$ for $s<0$,
whereas for $|s|\lesssim s_{\theta}$ the deviation is $O(1)$ and $\theta_{\epsilon}$ sweeps
continuously from $\pi$ to $0$ as $s$ crosses the origin. The two scales are therefore
well separated, $s_{\theta}\sim\epsilon\ll\sqrt{\epsilon}\sim s_{\phi}$, and this
separation splits the crossing into three regions.

The three regions are distinguished by which of $\theta$ and $\phi$ is in motion. In
region~$\mathrm{I}$ ($s_{\phi}\ll|s|\ll1$) the polar angle has already reached its pole
while the azimuth still sits at $0$; in region~$\mathrm{II}$
($s_{\theta}\ll|s|\ll s_{\phi}$) the azimuth turns from $0$ to $\alpha$ while
$\theta_{\epsilon}$ stays pinned at its pole; and in region~$\mathrm{III}$
($|s|\lesssim s_{\theta}$) $\theta_{\epsilon}$ sweeps from $\pi$ to $0$, the azimuth being
already fixed at $\alpha$.

The essential point is that the azimuthal turn of region~$\mathrm{II}$ takes place where
the regularized curve passes the two poles $\pm\vhat$, with $\theta$ settled at $0$ or
$\pi$. Evaluating equation~(\ref{eq:delta}) at the azimuthal scale $|s|\sim s_{\phi}$
gives $\delta\sim\epsilon u/(|\vv|\sqrt{\epsilon})=O(\sqrt{\epsilon})\to0$, so expanding
the monopole weight to second order in $\delta$, with $\theta_{\epsilon}=\delta$ on the
$s>0$ branch and $\theta_{\epsilon}=\pi-\delta$ on the $s<0$ branch,
\begin{equation}
1-\cos\theta_{\epsilon}=
\begin{cases}
1-\cos\delta=\tfrac12\delta^{2}+\cdots=O(\epsilon), & s>0,\\[2pt]
1+\cos\delta=2-\tfrac12\delta^{2}+\cdots=2-O(\epsilon), & s<0.
\end{cases}
\label{eq:weights}
\end{equation}
The same turn $\alpha$ is thus counted with weight $O(\epsilon)$ near one pole and weight
$2-O(\epsilon)$ near the other; this asymmetry is what produces the net $2\alpha$.

\subsection{Main theorem}
\label{sec:theorem}

\begin{theorem}[deviation of the solid angle]
\label{thm:main}
Let $\dvec(t)$ cross the origin transversally at $t_{0}$ with $\aperp\neq\mathbf0$. Let $C$
be the trajectory of the bare field $\nhat(t)$ (an open curve with antipodal endpoints
$\pm\vhat$), $\Omega[C]$ the value in Definition~\ref{def:omega}, and $\Omega(\epsilon\uhat)$
the solid angle of the regularized closed curve $\nhat_{\epsilon}$. If $\uhat\nparallel\vhat$
then
\begin{equation}
\Omega(+\epsilon\uhat)=\Omega[C]+2\alpha+O(\epsilon),
\label{eq:master}
\end{equation}
with $\alpha$ the azimuth~(\ref{eq:alpha}); equivalently
\begin{equation}
\gamma_{g}(+\epsilon\uhat)=\gamma_{g}[C]-\alpha+O(\epsilon),
\qquad \gamma_{g}[C]=-\tfrac12\Omega[C].
\label{eq:masterg}
\end{equation}
\end{theorem}

\begin{proof}
By Lemma~\ref{lem:bare}, the neighbourhood of the degeneracy point
contributes to $\Omega[C]$ an amount that vanishes with the window but is not zero for
any fixed window. We therefore compare the regularized and the bare solid angles window
by window rather than discarding this contribution. Fix $\delta_{0}$ with
$s_{\phi}\ll\delta_{0}\ll1$ and split both solid angles at $|s|=\delta_{0}$,
\begin{equation}
\Omega(\epsilon\uhat)=F_{\epsilon}+N_{\epsilon},\qquad \Omega[C]=F_{0}+N_{0},
\label{eq:FNsplit}
\end{equation}
where $F$ and $N$ denote the integrals of $(1-\cos\theta)\,\dd\phi$ over the far region
$|s|>\delta_{0}$ and the near window $|s|<\delta_{0}$, for the regularized and the bare
curve, respectively. We estimate the two differences $F_{\epsilon}-F_{0}$ and
$N_{\epsilon}-N_{0}$ in turn.

\emph{Far part.} On $|s|>\delta_{0}$ the field is bounded away from the origin,
$|\dvec|\gtrsim|\vv|\delta_{0}$, so $\nhat_{\epsilon}=\nhat+O(\epsilon)$ uniformly there;
the weight and the azimuth of the regularized curve then reduce to their bare
counterparts with $O(\epsilon)$ error, so that $F_{\epsilon}-F_{0}=O(\epsilon)$.

\emph{Near part.} The near window $|s|<\delta_{0}$ spans regions~$\mathrm{I}$,
$\mathrm{II}$, and~$\mathrm{III}$, which locate the difference $N_{\epsilon}-N_{0}$ as
follows. In region~$\mathrm{I}$ the bare curve term dominates the regularization in
equation~(\ref{eq:regsplit}), so in the integrand of $N_{\epsilon}-N_{0}$ the common bare
part cancels identically while the regularization perturbs the azimuthal measure
$\dd\phi$ and the weight only at $O(\epsilon)$; the region therefore contributes
$O(\epsilon)$ to the difference. In region~$\mathrm{III}$ the transverse vector of the
regularized curve reduces to $\epsilon\uhat_{\perp}$, the curve term
$\tfrac12|\aperp|s^{2}$ being only $O(\epsilon^{2})$ there, so its azimuth is constant,
$\phi_{\epsilon}=\alpha+O(\epsilon)$, and there is no azimuthal motion to weight despite
the $O(1)$ sweep of $\theta_{\epsilon}$; the bare integrand there is itself
$O(\epsilon)$, by the window estimate of Lemma~\ref{lem:bare} with
$w=s_{\theta}\sim\epsilon$.

The difference is therefore generated in region~$\mathrm{II}$, and it can be read off by
comparing the two factors of the integrand, $(1-\cos\theta)$ and $\dd\phi$. The azimuthal
factor supplies the same net turn on both branches: relative to the common bare
measure---which cancels in $N_{\epsilon}-N_{0}$---the regularization produces one turn of
size $\alpha$, executed on the $s<0$ side and undone on the $s>0$ side. Near the north
pole $\theta\approx0$, the returning turn of size $|\alpha|$ is counted with weight
$O(\epsilon)$ by equation~(\ref{eq:weights}), while the bare contribution over
region~$\mathrm{II}$ is bounded by the window estimate of Lemma~\ref{lem:bare} with
$w=s_{\phi}$---and in fact vanishes faster on this branch, the weight being $O(s^{2})$
there, integrating to $O(s_{\phi}^{3})=O(\epsilon^{3/2})$. Hence
\begin{equation*}
N^{s>0}_{\epsilon}-N^{s>0}_{0}=O(\epsilon),
\end{equation*}
the returning turn leaving no trace at this pole. Near the south pole
$\theta\approx\pi$, however, the weight is pinned at $2-O(\epsilon)$ throughout
region~$\mathrm{II}$; since regions~$\mathrm{I}$ and~$\mathrm{III}$ contribute only
$O(\epsilon)$, it factors out of the integral and each curve contributes twice its net
azimuthal change across the window,
\begin{equation*}
N^{s<0}=\int_{-\delta_{0}}^{0}(1-\cos\theta)\,\dd\phi
=2\big[\phi(0)-\phi(-\delta_{0})\big]+O(\epsilon).
\end{equation*}
The endpoint values are sharp: $\phi_{\rm bare}(0)=0$ by Lemma~\ref{lem:bare},
$\phi_{\epsilon}(0)=\alpha$ because the transverse vector at $s=0$ is exactly
$\epsilon\uhat_{\perp}$, and
$\phi_{\epsilon}(-\delta_{0})=\phi_{\rm bare}(-\delta_{0})+O(\epsilon)$ because
$\delta_{0}$ lies in region~$\mathrm{I}$. Subtracting the bare from the regularized
contribution,
\begin{equation*}
N^{s<0}_{\epsilon}-N^{s<0}_{0}
=2\big[\phi_{\epsilon}(0)-\phi_{\epsilon}(-\delta_{0})\big]
-2\big[\phi_{\rm bare}(0)-\phi_{\rm bare}(-\delta_{0})\big]+O(\epsilon)
=2\alpha+O(\epsilon),
\end{equation*}
the torsion term $\phi_{\rm bare}(-\delta_{0})$, common to both curves, dropping out.

\emph{Collection.} Collecting the far part and the three regions of the near window, the
total difference is
\begin{equation}
\Omega(\epsilon\uhat)-\Omega[C]
=\underbrace{O(\epsilon)}_{\Delta F}
+\underbrace{O(\epsilon)}_{\Delta N_{I}}
+\underbrace{2\alpha+O(\epsilon)}_{\Delta N_{II}}
+\underbrace{O(\epsilon)}_{\Delta N_{III}}
=2\alpha+O(\epsilon),
\label{eq:collect}
\end{equation}
where $\Delta F\equiv F_{\epsilon}-F_{0}$ and $\Delta N_{X}$ denotes the contribution of region~$X$ to $N_{\epsilon}-N_{0}$. This is equation~(\ref{eq:master});
equation~(\ref{eq:masterg}) follows on multiplying by the factor $-\tfrac12$ of
equation~(\ref{eq:gg}).
\end{proof}

\begin{remark}[mechanism]
\label{rem:mechanism}
Around the closed loop the azimuth returns to its initial value, so its net change
vanishes; the deviation $2\alpha$ is entirely a matter of \emph{where} the change takes
place. The turn imposed by the regularization is executed next to the Dirac string, where
the monopole weight is $2$, and undone at the opposite pole, where it vanishes, so the
weighted change is $2\alpha$---the finite analogue of the $O(w)$ boundary term in the
proof of Lemma~\ref{lem:bare}. For $\uhat_{\perp}\parallel\aph$, i.e.\ $\alpha=0$, the
regularization is aligned with the bare transverse direction and no turn is imposed at
all: the osculating-plane regularization is the unique neutral choice.
\end{remark}

\subsection{Corollaries}
\label{sec:corollaries}

\begin{corollary}[Berry's $\pi$ invariant]
\label{cor:pi}
For any $\uhat\nparallel\vhat$,
$\gamma_{g}(+\epsilon\uhat)-\gamma_{g}(-\epsilon\uhat)\to\pi\pmod{2\pi}$.
\end{corollary}

\begin{proof}
Reversing the regularization sends $\uhat_{\perp}\to-\uhat_{\perp}$, hence
$\alpha\to\alpha+\pi$ in~(\ref{eq:alpha}). By Theorem~\ref{thm:main},
$\gamma_{g}(+\epsilon\uhat)-\gamma_{g}(-\epsilon\uhat)
=(\gamma_{g}[C]-\alpha)-(\gamma_{g}[C]-\alpha-\pi)=\pi$; both $\alpha$ and $\gamma_{g}[C]$
cancel, so the value is independent of $\uhat$ and $\epsilon$.
\end{proof}

\begin{corollary}[reflection symmetry]
\label{cor:reflection}
If a plane $\Pi$ through the origin contains $\dvec(t)$ for all $t$ (reflection symmetry) and
the regularization is taken along the mirror normal $\uhat\perp\Pi$, then $\alpha=\pm\pi/2$,
$\gamma_{g}[C]=0$, and $\gamma_{g}(\pm\epsilon\uhat)=\mp\pi/2$.
\end{corollary}

\begin{proof}
Here $\vv,\av\in\Pi$, so $\aperp\in\Pi$ and $\Pi_{\rm osc}=\Pi$; since $\uhat\perp\Pi$ one has
$\uhat\cdot\vhat=0$, $\uhat_{\perp}=\uhat$, and $\Bh=\vhat\times\aph$ is the normal of $\Pi$,
so $\Bh=\pm\uhat$. Then $\uhat_{\perp}\cdot\aph=0$ and $\uhat_{\perp}\cdot\Bh=\pm1$, giving
$\alpha=\atantwo(\pm1,0)=\pm\pi/2$. As $\nhat$ is confined to the great circle
$\Pi\cap S^{2}$, Lemma~\ref{lem:geodesic} gives $\Omega[C]=0$, i.e.\ $\gamma_{g}[C]=0$;
Theorem~\ref{thm:main} then yields $\gamma_{g}(\pm\epsilon\uhat)=\mp\pi/2$.
\end{proof}

\begin{corollary}[non-universality of the average]
\label{cor:average}
$\tfrac12\big[\gamma_{g}(+\epsilon\uhat)+\gamma_{g}(-\epsilon\uhat)\big]=\gamma_{g}[C]-\alpha-\tfrac{\pi}{2}$,
so the average of the two regularized values depends on the regularization direction through
$\alpha$ and is not universal; under a reflection symmetry, where $\alpha=\pm\pi/2$ and
$\gamma_{g}[C]=0$, it equals $0$ or $-\pi$, agreeing with $\gamma_{g}[C]$ modulo $\pi$.
\end{corollary}

The contrast between Corollary~\ref{cor:pi} and Corollary~\ref{cor:average} is instructive:
the difference is universal because $\alpha$ cancels, whereas the average is not because
$\alpha$ survives; universality of the average would require the two half-loop contributions
to be exactly opposite, which only a mirror symmetry guarantees. A numerical verification of
Theorem~\ref{thm:main} and its corollaries, over several representative fields and across a
range of $\epsilon$, is given in~\ref{app:numerics}.

\subsection{Application: the Uhlmann phase selects the osculating-plane regularization}
\label{sec:uhlmann}

For the two-level Hamiltonian $H(t)=-\dvec(t)\cdot\sig$ the Gibbs state is
$\rho_{G}=\tfrac12(\mathbb 1+\boldsymbol r\cdot\sig)$ with $\boldsymbol r=\tanh(\beta|\dvec|)\,\nhat$.
Near the degeneracy $\tanh(\beta|\dvec|)\simeq\beta|\dvec|$, so by the expansion of $\dvec$,
$\boldsymbol r(s)\simeq\beta\,\vv\,s$: inside the Bloch ball $\boldsymbol r$ passes through
the origin along the straight line $\vhat$, without discontinuity. The pure state (Berry)
lives on the sphere and suffers the jump, whereas the mixed-state (Uhlmann) path is closed
and unbroken from the outset. This is the origin of the selection we now establish; it also
connects to the known behaviour of the Uhlmann phase for one-dimensional fermion systems and
for the Kitaev chain, where the state curve may pass through the maximally mixed
point~\cite{Viyuela2014,Andersson2016}.

The Uhlmann connection is
$A_{U}=-\tfrac{i}{2}f(r)\,(\nhat\times\dd\nhat)\cdot\sig$ with
$f(r)=1-\sqrt{1-r^{2}}$ and $r=\tanh(\beta|\dvec|)$. Writing
$\nhat=\operatorname{sgn}(s)\,\hat{\boldsymbol w}$ with $\hat{\boldsymbol w}$ smooth across
$s=0$, the sign squares out, $\nhat\times\dd\nhat=\hat{\boldsymbol w}\times\dd\hat{\boldsymbol w}$,
so the directional part of the connection is smooth and bounded through the degeneracy, while
$r\to0$ makes $f(r)\simeq\tfrac12 r^{2}\to0$. The central segment therefore contributes
nothing, and $\gamma_{U}$ is well defined without regularization for all $T>0$ and continuous
in $T$; geometrically, the line of sight from the centre stays on the straight segment
$\pm\vhat$ and sweeps no solid angle (cf.\ Lemma~\ref{lem:bare}). As $T\to0$, $f\to1$ and
$A_{U}\to-\tfrac{i}{2}(\nhat\times\dd\nhat)\cdot\sig$ becomes the pure-state Berry generator,
the central segment still contributing nothing, so
\begin{equation}
\gamma_{U}(T\to0)=\gamma_{g}[C]=-\tfrac12\Omega[C].
\label{eq:uhlmannlimit}
\end{equation}

Combining~(\ref{eq:uhlmannlimit}) with Theorem~\ref{thm:main} gives at once
\begin{equation}
\gamma_{U}(T\to0)=\gamma_{g}(+\epsilon\aph),\qquad
\gamma_{U}(T\to0)+\pi=\gamma_{g}(-\epsilon\aph),
\label{eq:select}
\end{equation}
i.e.\ the pure-state limit of the finite-temperature mixed state automatically selects the
osculating-plane regularization ($\alpha=0$); any other direction is offset by $-\alpha$. The
prescription ``close the open path in the way the physical path approaches the
endpoint''~\cite{GarzaSotoHagen2023} is thus realized here in closed form: the selected
geodesic is the osculating-plane great circle
$\operatorname{span}\{\dot\dvec(t_{0}),\ddot\dvec(t_{0})\}\cap S^{2}$, the curvature at the
degeneracy supplying the ``third point'', and the mixed-state construction realizes that
choice on its own. That the Uhlmann and Berry phases agree in this pure-state limit is an
instance of the Uhlmann--Berry correspondence~\cite{Uhlmann1986,Sjoqvist2000,Wang2023SciPost};
the correspondence is known to fail at genuine level degeneracies (Dirac
points)~\cite{Wang2025PRB}, which is consistent with the special role the degeneracy plays
throughout the present analysis.

\section{Conclusion}
\label{sec:conclusion}

We have analysed the geometric phase and solid angle of an open curve on $S^{2}$ with
antipodal endpoints---the generic situation when the control vector $\dvec(t)$ of a
two-level system crosses the origin transversally. The monopole
connection~(\ref{eq:A}) defines the open-path solid angle intrinsically
(Definition~\ref{def:omega}); every geodesic closing the antipodal endpoints contributes
zero to the phase line integral (Lemma~\ref{lem:antipodal}); and the degeneracy itself
leaves no concentrated contribution, despite the torsion-induced tilt of the transverse
azimuth (Lemma~\ref{lem:bare}). Yet exactly one closing geodesic realizes the identity
``geometric phase $=$ enclosed solid angle''. Theorem~\ref{thm:main} quantifies the
deviation of any other choice, $\Omega(\epsilon\uhat)=\Omega[C]+2\alpha+O(\epsilon)$, and
traces it to the asymmetry of the monopole weight between the two poles: the azimuthal
turn imposed by the regularization is counted with weight $2$ next to the Dirac string and
with weight $O(\epsilon)$ at the opposite pole (Remark~\ref{rem:mechanism}). Berry's $\pi$
invariant and the reflection values $\pm\pi/2$ follow as corollaries, and the
finite-temperature Uhlmann phase selects the osculating-plane regularization $\alpha=0$ in
its pure-state limit.

The practical content is a selection rule. To assign an enclosed solid angle---and hence a
geometric phase---to an open antipodal path, the path should be closed with the great
circle in the osculating plane of $\dvec(t)$ at the degeneracy, the unique geodesic
tangent to the curve at both endpoints; any other closing overcounts by $2\alpha$, with $\alpha$ computable from the local $2$-jet of the control curve---its velocity and acceleration at the crossing. This replaces the heuristic closing rules of the open-path literature~\cite{RakhechaWagh1996,GarzaSotoHagen2023} by a closed-form prescription, and
explains why they succeed when they do.

Several directions remain open. The analysis assumes a transversal crossing with
nonvanishing transverse curvature, $\aperp\neq\mathbf0$; inflectional crossings, where the
osculating plane degenerates, and trajectories with multiple crossings, where the
geodesics selected at successive degeneracies need not coincide, call for an extension of
Theorem~\ref{thm:main}. On the physical side, the $\pi$ invariant and the reflection
values $\pm\pi/2$ are directly testable in neutron and atom
interferometry~\cite{Wagh1998,Zhou2020}, and the interplay with mixed-state phases beyond
the Uhlmann construction---the interferometric phase of~\cite{Sjoqvist2000}, and the fate
of the Uhlmann--Berry correspondence at genuine level
crossings~\cite{Wang2025PRB}---deserves a systematic study.

\section*{Acknowledgments}
This work was supported by the Institute of Information \& Communications Technology Planning \& Evaluation (IITP) grant funded by the Korea government (MSIT) (RS-2022-II221029).

\appendix
\ifdef{\ClassIOPJournal}{}{\renewcommand{\thesection}{Appendix~\Alph{section}}}
\section{Numerical verification}
\label{app:numerics}

We verify Theorem~\ref{thm:main} and its corollaries on three
representative fields: Fields~A and~B are generic, whereas Field~C
has reflection symmetry:
\begin{align}
\text{Field A}&:\ \dvec=(\sin t,\;0.8\sin2t+0.3\cos3t-0.3,\;1-\cos t),\\
\text{Field B}&:\ \dvec=(0.9\sin t,\;0.5\cos2t-0.5,\;1.2(1-\cos t)+0.4\sin2t),\\
\text{Field C}&:\ \dvec=(\cos t,\;0,\;\sin t-1).
\end{align}
Fields~A and~B cross the origin transversally at $t_{0}=0$, and Field~C at
$t_{0}=\pi/2$; Field~C is confined to the mirror plane $\Pi$ (the $xz$-plane), and its
regularization is taken along the mirror normal $\uhat=\hat{\mathbf y}$.
For Field~A, $\vv=(1,1.6,0)$ and $\av=(0,-2.7,1)$, giving $\vhat=(0.530,0.848,0)$,
$\aph=(0.695,-0.434,0.573)$ and $\Bh=(0.486,-0.304,-0.820)$.

\emph{Method.} For each field the parameter $t$ is discretised on a uniform grid over a
period. The open-path solid angle $\Omega[C]$ is computed from the line
integral~(\ref{eq:defomega}), equivalently from the van Oosterom--Strackee triangle
sum~(\ref{eq:vos}); the regularized value $\Omega(\epsilon\uhat)$ is computed the same way for
the closed curve $\nhat_{\epsilon}$, with the polar and azimuthal angles
$(\theta_{\epsilon},\phi_{\epsilon})$ obtained by projecting onto the frame $(\aph,\Bh)$ and
applying $\atantwo$. The Uhlmann value $\gamma_{U}(T)$ is obtained by integrating the
connection of section~\ref{sec:uhlmann} at temperature $T$. Convergence is monitored by
halving $\epsilon$ from $10^{-1}$ to $10^{-3}$.

\emph{Main theorem.} For Field~A, $\Omega[C]=+0.4485\pi$, and at $\epsilon=0.002$ the
measured deviation $\Omega(+\epsilon\uhat)-\Omega[C]$ tracks $2\alpha$ across
directions---for $\uhat=\hat{\mathbf z},\hat{\mathbf x},(1,-0.6,0.5),\aph,\Bh$ the predicted
$2\alpha=-0.612\pi,\allowbreak +0.388\pi,\allowbreak +0.130\pi,\allowbreak 0,\allowbreak +1.000\pi$ against measured
$-0.610\pi,\allowbreak +0.386\pi,\allowbreak +0.129\pi,\allowbreak -0.001\pi,\allowbreak +0.997\pi$, respectively. Halving $\epsilon$ from
$0.02$ to $0.00125$ reduces the error from $0.015\pi$ to $0.001\pi$, by a factor $\approx1.8$
per halving, consistent with the $O(\epsilon)$ correction in equation~(\ref{eq:master}).

\emph{Three scales.} For $\epsilon=10^{-4}$, $\uhat=\hat{\mathbf z}$ the predicted transition
scale is $s_{\phi}=\sqrt{2\epsilon u/|\aperp|}=0.0107$; numerically $\phi_{\epsilon}\approx0$
for $|s|=0.1,0.03$, $\phi_{\epsilon}\approx\alpha$ for $|s|=0.001,0.0003$, and
$\phi_{\epsilon}\approx\alpha/2$ near $|s|\approx s_{\phi}$, while the weight is $2.000$ for
$s<0$ and below $1.4\times10^{-3}$ for $s>0$, confirming the weight asymmetry
of~(\ref{eq:weights}).

\emph{Corollaries.} At $\epsilon=0.002$ the difference
$|\gamma_{g}(+\epsilon\uhat)-\gamma_{g}(-\epsilon\uhat)|=0.995\pi,0.996\pi,0.997\pi$ for
$\uhat=\hat{\mathbf z},\hat{\mathbf x}$ and a tilted direction, confirming
Corollary~\ref{cor:pi}. For Field~C, $\gamma_{g}(\pm\epsilon\uhat)=\pm0.383\pi\to\pm0.489\pi$
as $\epsilon=0.2\to0.01$, converging to $\pm\pi/2$ (Corollary~\ref{cor:reflection}), with the
sum vanishing to machine precision.

\emph{Uhlmann selection.} At $T=0.02$ the values $\gamma_{g}(+\epsilon\aph)$ and
$\gamma_{U}$ agree to within $4\times10^{-4}\pi$ for Field~A, $10^{-4}\pi$ for Field~B, and
exactly for Field~C; halving $\epsilon$ from $0.01$ to $0.002$ reduces the difference from
$1.7\times10^{-3}$ to $4\times10^{-4}$, confirming equation~(\ref{eq:select}). The code reproducing
these data is available from the authors.

\ifdef{\ClassIOPJournal}{\bibliographystyle{iopart-num}}{\bibliographystyle{unsrt}}
\bibliography{OpenPath_S2_refs}

\end{document}